\documentclass[11pt,a4paper,english]{article}
\usepackage[T1]{fontenc}
\usepackage[utf8]{inputenc}
\usepackage{authblk}
 \title{A Poincar\'e Covariant Noncommutative Spacetime?}
\author[1]{Albert Much\footnote{amuch@matmor.unam.mx} }
\author[2]{Jos\'e David  Vergara \footnote{vergara@nucleares.unam.mx}}
\affil[1]{Centro de Ciencias Matem\'aticas\\
	UNAM\\
  Morelia, Michoac\'an, Mexico}  
 \affil[2]{Instituto de Ciencias Nucleares, UNAM,  D.F. , M\'exico }

\usepackage{float}            
\usepackage[T1]{fontenc}   
\usepackage[cmex10]{amsmath}  
\usepackage{amstext}          
\usepackage{mathrsfs}
\usepackage{amsfonts}         
\usepackage{amssymb}          
\usepackage{amsthm}    
\usepackage{setspace}
\usepackage{bm}               
\usepackage{enumerate}        
\usepackage[bottom]{footmisc} 
\usepackage{array}            
\usepackage{textcomp}         
\usepackage{pdfpages}         
\usepackage{parskip}          
\usepackage[right]{eurosym}   
\usepackage{xcolor}           
\usepackage[square,numbers]{natbib}	  

\usepackage{amsmath}  
\usepackage[colorlinks,pdfpagelabels,pdfstartview = FitH,bookmarksopen = true,bookmarksnumbered =
true,linkcolor = black,plainpages = false,hypertexnames = false,citecolor = black]{hyperref}

\newtheorem{theorem}{\textsc{Theorem}}[section]
\newtheorem{lemma}{\textsc{Lemma}}[section]
\newtheorem{proposition}{\textsc{Proposition}}[section]

\theoremstyle{definition}
\newtheorem{definition}{\textsc{Definition}}[section]

\newtheorem{convention}{Conventions}[section]

\theoremstyle{remark}
\newtheorem{remark}{Remark}[section]

\newcommand{\R}{\mathbb{R}}

\newcommand{\bone}{\mathbf{1}}

\usepackage{fancyhdr}
\numberwithin{equation}{section} 
\begin{document}
	\maketitle
	\abstract{We interpret, in the realm of relativistic quantum field theory, the tangential operator given by Coleman, Mandula \cite{CM} (see also \cite{Muc8}) as an appropriate coordinate operator. The investigation shows that the operator generates a Snyder-like noncommutative spacetime with a  minimal length that is  given by the mass. By using this operator to define a noncommutative spacetime, we  obtain a Poincar\'e invariant noncommutative spacetime and in addition solve the soccer-ball problem. 	 Moreover, from recent progress in deformation theory we extract the idea how to obtain, in a physical and mathematical well-defined manner, an emerging noncommutative spacetime. This is done by a strict deformation quantization known as Rieffel deformation (or warped convolutions).  The result is a noncommutative spacetime combining a Snyder and a Moyal-Weyl type of noncommutativity that in addition behaves covariant under transformations of the \textbf{whole} Poincar\'e group.} 
	\tableofcontents

\begin{convention}
	We use $d=n+1$, for $n\in\mathbb{N}$ and the Greek letters are split into  $\mu, \,\nu=0,\dots,n$. Moreover, we use Latin letters for the spatial components which run from $1,\dots,n$ and we choose the following convention for the Minkowski scalar product of $d$-dimensional vectors, $a\cdot b=a_0b^0+a_kb^k=a_0b^0- \vec{a}\cdot\vec{b}$. Furthermore, we use the common symbol $\mathscr{S}(\mathbb{R}^{n,d})$ for the Schwartz-space and $\mathcal{H}$ for a Hilbert space.
\end{convention} 
\section{Introduction}  
Pick up an arbitrary event $\mathscr{A}$, whose perpendicular displacement from the center-of-mass world-line is $X_{\mu}$,
then the angular momentum, about $\mathscr{A}$ is 
\begin{equation*}
J_{\mu\nu}=P_{\mu}X_{\nu}-P_{\nu}X_{\mu}.
\end{equation*}
For the moment we neglect the intrinsic angular momentum. Let the observable $P^{\mu}$ denote the relativistic four-momentum. By knowing the angular momentum of $\mathscr{A}$  and the four-momentum,
one can calculate the vector from $\mathscr{A}$ to the center-of-mass world-line, \cite[Chapter II.5.11, Box E]{MG}
\begin{equation*}
X_{\mu}=\frac{J_{\mu\nu}P^{\nu}}{M^2}
\end{equation*}where here	$M$ denotes the mass.  	It is interesting to note that the author in  \cite{PR48} found a similar operator while searching  for a relativistic generalization of the Newtonian center-of-mass world-line. The operator transforms covariantly under Lorentz transformations and it is noncommutative    among its own components. \newline\newline
Nowadays, investigations prove the physical and/or mathematical necessity for 
spacetimes that carry a noncommutative structure, \cite{DFR},  \cite{AA}, \cite{Muc1} and \cite{Muc9}. One of the most common studied models is the so called Moyal-Weyl spacetime, generated by self-adjoint coordinate operators $\hat{x}$ that fulfill,
$$[\hat{x}_{\mu},\hat{x}_{\nu}]=i\Theta_{\mu\nu},$$
where $\Theta$ is a constant skew-symmetric matrix. Mathematically, it can be argued that this spacetime is a generalization of the commutative coordinate structure such as the quantum-mechanical symplectic structure is between the momentum and the coordinate.  Physically, these commutation relations are consequences of the fact  that localization with high accuracy at Planck scale causes a gravitational collapse, \cite{DFR}. However, one major setback of such spacetimes is the explicit breaking of Lorentz-covariance. \newline \newline In our opinion, this is due to the fact that algebras are imposed a priori without the existence of physically motivated operators (observables) that generate these spacetimes.  Hence,  in this letter we take the relativistic definition of the center-of-mass world-line as an appropriate coordinate operator and  investigate its consequences. First we define it   formally hermitian,  
\begin{equation}\label{cop}
X_{\mu}=\frac{1}{2{M^2}}({J_{\mu\nu}P^{\nu}}+{P^{\nu}J_{\mu\nu}}),
\end{equation} 
where $M^2$ denotes  the mass operator that is given by  
\begin{equation}\label{em}
M^2=P_{0}P^{0}+P_{i}P^{i}=P_{\mu}P^{\mu}.
\end{equation} The operator which we call the center-of-mass world-line is essentially what is known in literature   as the  tangential derivative or the Coleman-Mandula operator. It maintains the  mass shell condition that is imposed by the 	Klein-Gordon equation, see \cite[Equation (13)]{CM} and \cite[Equation (5)]{Muc8} and it is given (in momentum space) by,
\begin{align*}
i\nabla_{\mu}=i\left( 
 \frac{\partial}{\partial p^{\mu}}-m^{-2}p_{\mu}p^{\nu}\frac{\partial}{\partial p^{\nu}}
\right).
\end{align*}
 The agreement of the center-of-mass operator with the Coleman-Mandula derivative is seen by rewriting the angular momentum in the momentum space  as a product of the momentum and derivatives w.r.t. the momentum.\\\\
 What is the deeper insight behind the idea that the Coleman-Mandula operator is in a relativistic context a rightful coordinate operator? Well, from considerations in quantum mechanics in the coordinate space we know that the momentum operator acts as the tangential operator to the coordinate. On the other hand in momentum space the coordinate operator is the tangential operator to the momentum. Since,   in a relativistic particle theory we have the restriction of the mass shell condition on the momentum space, that comes from the Klein-Gordon Equation, this space is   restricted and cannot have the same tangential operator as for the non-restricted QM-case.
\\\\ In more mathematical terms, in the quantum mechanical case the momentum and the coordinate   are represented by unbounded self-adjoint operators  with the spectrum being all of the real vector space. Hence the eigenvalues go from minus infinity to plus infinity. However, in the relativistic case the momentum   which is a four-vector,   although a self-adjoint operator, has to  fulfill the so called spectrum-condition. This   means that  the spectrum of energy-momentum operator lies in the forward light cone.  Hence, the tangential derivative of this particular operator cannot be simply the derivative. This is due to the fact that the spectrum is not the whole real space as in the non-relativistic quantum mechanical case.   Therefore, the Newton-Wigner operator is not the tangential operator but the correct one is rather the Coleman-Mandula operator.\newline\newline
Next, we take a closer look at the algebra that the coordinate operator satisfies with   generators of the Poincar\'e group. Beforehand, we give the commutator relations of the Poincar\'e algebra 
\begin{align}\label{crpj}
[P_{\mu},P_{\nu}]=0,\qquad[J_{\rho\sigma},P_{\mu}]=i (\eta_{\mu\rho}P_{\sigma}-\eta_{\mu\sigma}P_{\rho} ), 
\end{align} 
\begin{align}\label{crpj1}
[J_{\mu\nu},J_{\rho\sigma}]=i\left(
\eta_{\mu\rho}J_{\nu\sigma }-\eta_{\mu\sigma}J_{\nu\rho }-\eta_{\nu\rho}J_{\mu\sigma }+\eta_{\nu\sigma }J_{\mu\rho }
\right).
\end{align} 

 	First, we  consider the commutator of the Coleman-Mandula operator  with the momentum.
\begin{lemma}
	By taking the relativistic definition of the center-of-mass world-line we obtain the following commutator relation with the momentum operator 
\begin{equation}\label{cr1}
[X_{\mu},P_{\nu}]=i\biggl(\eta_{\mu\nu}-\frac{P_{\mu}P_{\nu}}{M^2}\biggr).
\end{equation} 
\end{lemma}
\begin{proof}  
	Hence, by using the Poincar\'e algebra the commutator is explicitly calculated,
	\begin{align*}
	[X_{\mu},P_{\nu}]&=\frac{1}{2{M^2}}[({J_{\mu\rho}P^{\rho}}+{P^{\rho}J_{\mu\rho}}),P_{\nu}]\\&
	=\frac{1}{2{M^2}}[ 2{J_{\mu\rho}P^{\rho}}+ [P^{\rho},J_{\mu\rho}]   ,P_{\nu}]
	 \\& 
			=\frac{i}{ {M^2}}(\eta_{\mu\nu}P_{\rho}-\eta_{\rho\nu}P_{\mu}
			) P^{\rho},
	\end{align*}
	where in the last lines we used the definition of the mass operator (see Formula (\ref{em})).
\end{proof}
\begin{remark}
	In order to keep track of the dimensions and the physical constants we write the former commutator relation without the usual convention $\hbar=c=1$, e.g. $$
	[X_{\mu},P_{\nu}]=i\hbar\biggl(\eta_{\mu\nu}-\frac{P_{\mu}P_{\nu}}{M^2c^4}\biggr). $$ Hence, the extra term acts as a relativistic correction to the uncertainty principle. For negligible spatial momentum ($c\rightarrow\infty$), as is the case in non-relativistic quantum mechanics, we obtain the common Heisenberg relations.
	For $\mu,\nu=0$ and negligible spatial momentum the commutator is zero, therefore it respects Pauli's theorem about the so called "time" operator. Furthermore, notice that the modified commutator relations between the coordinate and the momentum operator, given in Equation (\ref{cr1}) are  equivalent to considering  a momentum dependent metric (see \cite{DV} and references therein),  i.e. 
\begin{equation*}
g_{\mu\nu}(P) = \eta_{\mu\nu} -\frac{P_\mu P_\nu}{M^2 c^4}.
\end{equation*}
\end{remark}
The spatial part of this commutator relation is the origin of the quantum mechanical uncertainty principle with an additional term. The extra term can be  physical motivated although it was not imposed from the get-go. 	In a series of papers  \cite{Mag}, \cite{Mag2}  and \cite{Mag3}, generalized quantum mechanical commutation relations were considered  motivated by quantum-gravity principles. 
These generalized uncertainty principles are exactly of the form that we obtained by promoting the center-of-mass   world-line to an operator, see Equation (\ref{cr1}). Due to the additional term, w.r.t. the common canonical commutation relations   one is led to a minimal length (see \cite{Mag}, \cite{Mag2}  and \cite{Mag3}), which is the  principle  idea of noncommutative geometry applied to  physics.
In particular, the spatial part of the operator $X$ can be considered as a generalization or deformation of the standard position operator given in QFT, known by the name of Newton-Wigner-operator (see \cite{NW49}, \cite{PR48}, \cite{J80}, \cite{SS} and \cite{Muc3}). This statement is further clarified in the following sections. 
\newline\newline  Note that is was proven in the aforementioned papers  that a  noncommutative spacetime emerges if and only if the coefficient in front of the extra-term in the generalized commutation relations is positive. This statement would coincide in our case with the   requirement  	$M^{2}>0$, which is a restriction with well-implied physical meaning.  

\begin{lemma} \label{le1}
	The following commutation relations prove  that the   center-of mass operator,  transforms covariant
	\begin{equation}\label{e1}
	[J_{\rho\sigma},X_{\mu}]= i\biggl(\eta_{\mu\rho}X_{\sigma}-\eta_{\mu\sigma}X_{\rho}\biggr) 
	\end{equation} 
	and the non-commutativity of this operator generates a particular \textbf{Snyder-spacetime},
	\begin{equation}\label{ss}
	[X_{\mu},X_{\nu}]= i\frac{J_{\mu\nu}}{M^2}.
	\end{equation} 
\end{lemma}
\begin{proof}
The proof is a straight forward calculation  and in order to see that the Coleman-Mandula operator transforms covariant we first calculate the commutator with the Lorentz generator, i.e. 
	\begin{align*}
	2{M^2}[J_{\rho\sigma},X_{\mu}]&=  [J_{\rho\sigma},({2J_{\mu\nu}P^{\nu}}+{[P^{\nu},J_{\mu\nu}]})]\\&= 
	\left(
	2[J_{\rho\sigma},J_{\mu\nu}]P^{\nu}+2J_{\mu\nu}[J_{\rho\sigma},P^{\nu}]
	+3i 
	[J_{\rho\sigma}, P_{\mu}]\right) 
	\\&=i
	\biggl(  -
	2\left(
	\eta_{\mu\rho}J_{\nu\sigma }P^{\nu}-\eta_{\mu\sigma}J_{\nu\rho }P^{\nu}- J_{\mu\sigma }P_{\rho}+ J_{\mu\rho }P_{\sigma}
	\right)
	\\&+2  
	(J_{\mu\rho} P_{\sigma}-J_{\mu\sigma} P_{\rho} ) 
	+3i  (\eta_{\mu\rho}P_{\sigma}-\eta_{\mu\sigma}P_{\rho} ) \biggr)
	\\&=i
	\biggl(   
	2\left(	\eta_{\mu\rho}J_{\sigma\nu }P^{\nu}
	-\eta_{\mu\sigma}J_{\rho \nu }P^{\nu} 
	\right) 
	+3i  (\eta_{\mu\rho}P_{\sigma}-\eta_{\mu\sigma}P_{\rho} ) \biggr).
	\end{align*}
	Next we turn to the commutator of different components of the  center-of mass operator,
	\begin{align*}
	4{M^4}[X_{\mu},X_{\nu}]&=  [ 2J_{\mu\rho}P^{\rho}+ 3iP_{\mu} , 2J_{\nu\sigma}P^{\sigma} + 3iP_{\nu} ]\\&
	= 4 [  J_{\mu\rho}P^{\rho}, J_{\nu\sigma}P^{\sigma}]+ 6i\bigl(
	[  J_{\mu\rho}P^{\rho}, P_{\nu}]-\mu\leftrightarrow\nu\bigr) 
	\\& 
	= 4 [  J_{\mu\rho}, J_{\nu\sigma}]P^{\sigma}P^{\rho}+
	4 J_{\nu\sigma}[  J_{\mu\rho}, P^{\sigma}]P^{\rho}+
	4 J_{\mu\rho}[  P^{\rho}, J_{\nu\sigma}]P^{\sigma}    
	\\&+  
	6i \bigl(
	[  J_{\mu\rho}, P_{\nu}]P^{\rho}-\mu\leftrightarrow\nu\bigr) 
	\\&=4 i\left( - J_{\rho\nu }P^{\rho}P_{\mu}- J_{\mu\rho }P^{\rho}P_{\nu}+ J_{\mu\nu }P_{\rho}P^{\rho}
	\right)\\&+
	4i J_{\nu\sigma}
	(\eta_{\mu}^{\sigma}P_{\rho}  P^{\rho}-P^{\sigma}P_{\mu})
	-
	4i J_{\mu\rho}
	( \eta_{\nu}^{\rho}P_{\sigma}   P^{\sigma}   -P^{\rho}P_{\nu})
	-
	6i \bigl(  [P_{\mu},P_{\nu}]\bigr) 
	\\& =  + 4i
	(J_{\mu\nu}P_{\rho}  P^{\rho}) .
	\end{align*}\end{proof} 
This type of algebra was already noticed in $1947$ by Snyder, see \cite{SN}, and it  is considered as the first work on noncommutative  spacetimes,  that in addition were implemented in a Lorentz invariant fashion. 
In the work of  Snyder the mass did not play the major role that it plays in our case. If one would interchange in our formulas of the mass operator $M$ with  	$1/a$ where $a$ is a natural length one would obtain the Snyder algebra. So, in an intuitive sense, the physical implication of our coordinate operator  is the following: the natural length of non-commutativity depends on the mass of the particle under observation. This in particular means that the heavier an object is the less one notices the 
underlying noncommutative spacetime. Let us at this point mention the de Broglie wavelength which is given by $\lambda= {{\hbar}}/{p}$. The de Broglie formula tells us that a particle with small mass  behaves more like a wave as a particle, with wavelength given by \(\lambda\). 
In particular this means that the lighter a particle is the more it will act as a QM object. One has to point out that the QM nature of a particle depends on $\hbar$ (Planck constant) and of the mass of the particle (and of course it also depends on the velocity which we leave out for the moment).    
 \\\\
Our coordinate operators have additional properties which distinguish them from the generators obtained by Snyder. 
\begin{lemma}\label{l13}
The mass operator $M^2=P_{\mu}P^{\mu}$ is a Casimir for the whole algebra, i.e. for $(X,P,J)$.
\end{lemma}
\begin{proof}
The proof is straight forward  since $X$ is a product of the Poincar\'e group for which the mass operator is the central element.  
\end{proof}After we introduced the algebra and the motivation for the Coleman-Mandula operator the rest of the paper is organized as follows;  Section two represents the operator in terms of the free scalar field and further on we prove the essential self-adjointness of the operator. Furthermore, we prove the Poincar\'e-covariant transformation behavior of the    spacetime that is generated by the Coleman-Mandula operator. The third section is concerned with the deformation of the  Coleman-Mandula  operator, where first we prove that the deformation is mathematically well-defined and  afterwards calculate the explicit result  of the deformation. We conclude the third section by proving, as in the undeformed case, the covariant transformation of the noncommutative spacetime that is generated by the deformed Coleman-Mandula operator. 

\section{Properties of the Coleman-Mandula Operator}
"No deformation without representation!", \cite{PI}. In order to perform a deformation of the Coleman-Mandula operator we need to choose a representation.  The choice of representation is the free massive scalar field. Hence, in the following subsection we  express the Coleman-Mandula operator in the Fock-space of the free field. This is done    by rewriting the formally hermitian product of the Poincar\'e generators in the momentum space and then  expressing those in terms of the Fock-space operators. The plan to proceed in this section is the following; After introducing the Fock-space briefly we write the conjugate operator $X_{\mu}$ for a one-particle subspace. Next, we show that the operator is self-adjoint, hence the identification with a generalized observable, as the coordinate, can be justified from the point of view of real eigenvalues. 
\newline\newline
Since we  use the scalar field to give a representation of the Coleman-Mandula operators
in terms of explicit Fock-space operators, we give a short summary of the Fock-space and  the operators that are used in the subsequent sections. 
\subsection{Bosonic Fock-space}\label{sbfs} The four-dimensional     Fock-space for a free massive scalar field is defined as follows. A field with momentum $\mathbf{p} \in \mathbb{R}^3$ has the energy $\omega_{\mathbf{p}}=+\sqrt{\mathbf{p}^2+m^2}$. Moreover the Lorentz-invariant measure is given by   $d\mu(\mathbf{p} )=d^3\mathbf{p}/( {2\omega_{\mathbf{p}}})$.
\begin{definition}\label{bf}
	The \textbf{Bosonic Fock-space} $\mathscr{F^{+}({H})}$ is defined 
	as in \cite{BR}:
	\begin{equation*}
	\mathscr{F^{+}({H})}=\bigoplus_{k=0}^{\infty}\mathscr{H}_{k}^{+},
	\end{equation*}
	where $\mathscr{H}_{0}=\mathbb{C}$ and  the symmetric $k$-particle subspaces are given as
	\begin{align*}
	\mathscr{H}_{k}^{+} =\{\Psi_{k}: \underbrace{	H^{+}_{m}  \times  \dots \times 
		H^{+}_{m}}_{k-times} \rightarrow \mathbb{C}\quad \mathrm{symmetric}
	| \left\Vert  \Psi_k \right\Vert^2  <\infty\},
	\end{align*}
	with $
	H^{+}_{m}:=\{p\in \mathbb{R}^{4}|p^2=m^2,p_0>0\}.$
\end{definition}
The  annihilation and creation operators $a,a^{*}$ of the   
Fock-space satisfy,  
\begin{align}\label{pccr}
[a_c(\mathbf{p}), a_c^{*}(\mathbf{q})]=2\omega_{\mathbf{p} 
}\delta^3(\mathbf{p}-\mathbf{q}), \qquad
[a_c(\mathbf{p}), a_c(\mathbf{q})]=0=[a_c^{*}(\mathbf{p}), a_c^{*}(\mathbf{q})].
\end{align} 
By using  $a_c,a_c^{*}$   the particle number operator and the momentum operator are  defined in the following manner.
\begin{equation}\label{pcaopm}
N=\int d\mu(\mathbf{p}) a_c^{*}(\mathbf{p})a_c(\mathbf{p}), \qquad P_{\mu}=\int
d\mu(\mathbf{p})\,  p_{\mu} a_c^{*}(\mathbf{p})a_c(\mathbf{p}),
\end{equation}
where  $p_{\mu}=(\omega_{\mathbf{p}}, \mathbf{p})$.  
The generators  of the proper orthochronous Lorentz group $\mathscr{L}^{\uparrow}_{+}$ are given in terms of the  covariant ladder operators   by, \cite[Equation
3.54]{IZ} and in this context see also \cite[Appendix]{SS}
\begin{align}\label{lbcaop1} 
J_{j0}&= -i\int d\mu(\mathbf{p})\,{a}_c^{*}(\textbf{p}) \,\big( \omega_{\textbf{p}}\frac{\partial}{\partial p^j } \big)  a_c(\textbf{p}),
\\ \label{lbcaop2}
J_{ik}&=i \int d\mu(\mathbf{p})\,  {a}_c^{*}(\textbf{p}) 
\left(p_i \frac{\partial}{\partial p^k }-p_k\frac{\partial}{\partial p^i }\right)
a_c(\textbf{p}),
\end{align}
where $J_{0j}$ denotes the boost operators and $J_{ik}$ rotations. 

\subsection{Representation of the Coleman-Mandula Operator in Fock-space}
In this section we first represent the Coleman-Mandula operator by using the massive scalar field.  Furthermore, we prove the essential self-adjointness of this operator in  order to identify the Coleman-Mandula operator with an observable. Subsequently, we calculate the action of the Poincar\'e group on the Coleman-Mandula operator.
We proceed with the first theorem concerning the concrete representation of the generalized coordinate operator. 
\begin{theorem}
	The representation of the temporal part of the Coleman-Mandula operator $X_{\mu}$ is given on a one-particle  wave-function $\varphi\in\mathscr{S}(\R^3)$ of the massive scalar field as,
	\begin{equation}\label{cop0}
	(X_{0}\varphi)(\mathbf{p})=-\frac{i}{m^2}\omega_{\textbf{p}}\left(\frac{3}{2}+ p^k\frac{\partial}{\partial p^k}
	\right)\varphi(\mathbf{p}),
	\end{equation}
while the spatial part of the operator is given by,
	\begin{equation}\label{xrm}
		(X_{j}\varphi)(\mathbf{p})=+i\bigg(
	\frac{3 \,p_{j}}{2m^2} + 
	\frac{p_jp^{k}}{ m^2}	  \frac{\partial}{\partial p^k }- \frac{\partial}{\partial p^j} \bigg)\varphi(\mathbf{p}).
		\end{equation}
	
	 Moreover,  the Coleman-Mandula operator $X_{\mu}$ is  hermitian on the domain $\mathscr{S}(\R^3)$ w.r.t. to the Lorentz-invariant measure, 
	  i.e. w.r.t. the inner product on 	$\mathscr{H}_{1}^{+},$
	  \begin{equation}
	  \langle \psi, X_{\mu}\varphi\rangle= 	\langle X_{\mu}\psi, \varphi\rangle.
	  \end{equation}
\end{theorem}
\begin{proof}
	The calculation is straight-forward by using the concrete form of the Poincar\'e generators which are given explicitly in Equations (\ref{pcaopm}) and (\ref{lbcaop1}). We first calculate the zero-component,
	\begin{align*}
(	X_0\varphi)(\mathbf{p})&=\frac{1}{2{m^2}}\big(({J_{0\nu}P^{\nu}}+{P^{\nu}J_{0\nu}}\big)\varphi)(\mathbf{p})
  =\frac{1}{2{m^2}}\big(({ [J_{0\nu},P^{\nu}]}+{2P^{\nu} J_{0\nu} }\big)\varphi)(\mathbf{p})
\\&=\frac{1}{2{m^2}}\big((-{3iP_0}+{2P^{k}J_{0k}}\big)\varphi)(\mathbf{p}) = -\frac{i}{m^2} \omega_{\textbf{p}} \left(  
 \frac{3}  {2 }  +   p^k  
  \frac{\partial}{\partial p^k } \right) \varphi (\mathbf{p})
	\end{align*}
	where in the last lines we used the algebra of the Poincar\'e group, their explicit form acting on a momentum vector state, 
	\begin{align*}
	(	X_j\varphi)(\mathbf{p})&=\frac{1}{2{m^2}}\big(({J_{j\nu}P^{\nu}}+{P^{\nu}J_{j\nu}}\big)\varphi)(\mathbf{p})
 =\frac{1}{2{m^2}}\big(\big({ [J_{j\nu},P^{\nu}]}+2{P^{\nu} J_{j\nu} }\big)\varphi\big)(\mathbf{p})
	\\&=\frac{1}{2{m^2}}\big(( {3iP_j}+{2{P^{\nu} J_{j\nu} }}\big)\varphi)(\mathbf{p}) 	\\&=\frac{i}{m^2} p_{j}
	\left( 
	\frac{3  }{2}  +  p^{k}	  \frac{\partial}{\partial p^k }\right)\varphi(\mathbf{p})-i      \frac{\partial}{\partial p^j }  \varphi(\mathbf{p}) .
	\end{align*}
	Next we turn to the proof of hermiticity and start by considering the spatial part of the operator $X$. Moreover we split the operator into three parts,
		\begin{align*}\left(
-	X_j\varphi\right)(\mathbf{p}) =-i\bigg(
	\frac{3 \,p_{j}}{2m^2}  + \frac{p_jp^{k}}{m^2}	  \frac{\partial}{\partial p^k }- \frac{\partial}{\partial p^j} \bigg)\varphi(\mathbf{p}):=\left(  X_j^1\varphi+X_j^2\varphi+X_j^3\varphi\right)(\mathbf{p}).
		\end{align*}
	The considerations start with the second term, i.e.  
		\begin{align*}
		\langle \psi, X^2_{j}\varphi\rangle&=
		\int d\mu(\mathbf{p})\overline{\psi}(\mathbf{p})	
		(	X_j^2\varphi)(\mathbf{p})\\&=
		-\frac{i}{m^2} \int d\mu(\mathbf{p})\overline{\psi}(\mathbf{p})	
		\left(  p_jp^{k}	  \frac{\partial}{\partial p^k }\right)\varphi(\mathbf{p})  
		\\&=
		 \frac{i}{m^2} \int d\mu(\mathbf{p})\,\omega_{\textbf{p}}	\left( \frac{\partial}{\partial p^k }	\left(  \frac{p_jp^{k}}{\omega_{\textbf{p}}}	 \overline{\psi}(\mathbf{p})\right)	\right)
\varphi(\mathbf{p})  	\\&=
		\frac{i}{m^2} \int d\mu(\mathbf{p})\,\left(  \left( 
	3    + \frac{m^2}{\omega_{\textbf{p}}^2}+ { p^{k}}  \frac{\partial}{\partial p^k }
		\right)\overline{\psi}(\mathbf{p})	\right)
		\varphi(\mathbf{p}) ,
		\end{align*}
		where in the last lines a partial integration was performed. For the first term we have 
	\begin{align*}
	\langle \psi, X^1_{j}\varphi\rangle&=
	\int d\mu(\mathbf{p})\overline{\psi}(\mathbf{p})	
	(	X_j^1\varphi)(\mathbf{p}) =
i	\int d\mu(\mathbf{p})\overline{\psi}(\mathbf{p})	
\frac{\partial}{\partial p^j }\varphi(\mathbf{p})\\&=-
i	\int d\mu(\mathbf{p}) \,\omega_{\textbf{p}} \frac{\partial}{\partial p^j }	\left(  \frac{1}{\omega_{\textbf{p}}}	 \overline{\psi}(\mathbf{p})\right)
 \varphi(\mathbf{p})\\&=-
 i	\int d\mu(\mathbf{p}) \, 	\left(\left(  \frac{p_j}{\omega_{\textbf{p}}^2}	 +\frac{\partial}{\partial p^j }\right)\overline{\psi}(\mathbf{p})\right)
 \varphi(\mathbf{p}).
		\end{align*}
By adding all the terms hermiticity follows. Next, we turn to the zero component and consider the second term containing the derivative, which we denote by $X_0^2$,

		\begin{align*}
			\int d\mu(\mathbf{p})\overline{\psi}(\mathbf{p})	
		(	X^2_0\varphi)(\mathbf{p}) &= -\frac{i}{m^2}
		\int d\mu(\mathbf{p})\overline{\psi}(\mathbf{p})	\,
		 \omega_{\textbf{p}} \,
  p^k  
		 \frac{\partial}{\partial p^k } \varphi (\mathbf{p})\\&=
		 \frac{i}{m^2}
		 \int d\mu(\mathbf{p})\,	 \omega_{\textbf{p}} \left(\big(3+
		 p^k  
		 \frac{\partial}{\partial p^k } \big)\overline{\psi}(\mathbf{p})\right)	\,
	\varphi (\mathbf{p}),
		\end{align*}
	 where we performed a partial integration and used the fact that the measure comes with an energy dependent term.  	Another, equivalently valuable,  path to prove the concrete form or the hermiticity of the operator $X$ is   by representing the operators of the Poincar\'e group in their coordinate representation and evaluating them on the scalar product 
	\begin{equation}\label{spx}
	(\Phi,\Psi)_t=i\int_t d^3x \,\, \overline{\Phi (x)}\overleftrightarrow{ \partial_{0}}\Psi (x),
	\end{equation} 
	and afterwards performing a mass-shell Fourier transformation of the functions, i.e. 
	\begin{align*}
	\Phi (x)&=\langle x| \Phi\rangle = \int d\mu(\mathbf{p}) e^{-ipx}\langle p| \Phi\rangle=
	\int d\mu(\mathbf{p}) e^{-ipx}  \tilde{\Phi}(p).
	\end{align*}
	However, since the two different approaches lead to the same result,   we chose for the proof the more intuitive and efficient  momentum representation.
\end{proof}
 \begin{remark}
We recover the algebra of the Coleman-Mandula and momentum operator by using the explicit representation in the one-particle subspace. Moreover,    the authors in \cite{LS1},  \cite{LS2} reexamined the Snyder algebra and obtained a similar representation  for the spatial part (see \cite[Equation 3.1]{LS2}) as we did in Equation (\ref{xrm}). However, there are  considerable differences. The first difference is the fact that they use a  free parameter that is responsible for  the strength of the deformation of the usual commutator relations, i.e.
\begin{equation} 
[X_{\mu},P_{\nu}]=i\biggl(\eta_{\mu\nu}-\frac{P_{\mu}P_{\nu}}{\Lambda^2}\biggr) 
\end{equation}  
 while in our work we identified the parameter $\Lambda$ with the mass of the particle. Due to this deformation parameter their representation of the spatial part of the vector is given w.r.t. a different measure, while we simply use the well-known Lorentz-invariant  measure from QFT (or relativistic QM). Moreover, in \cite{LS2} they consider the two possibilities of    the deformation parameter $\Lambda^2$ being smaller or bigger than zero. Since we use positive energy representations we only have the positive case. 
 \end{remark} 
Next, we investigate the self-adjointness of the Coleman-Mandula operator. From a physical point of view this is important in order to guarantee real eigenvalues, i.e. it is conform with the quantum mechanical axiom that observables are represented by self-adjoint operators.
\begin{theorem}
The Coleman-Mandula operator $X_{\mu}$ is essentially self-adjoint  on the domain $\mathscr{S}(\R^3)$ w.r.t. to the Lorentz-invariant measure of the Hilbert space	$\mathscr{H}_{1}^{+}$ (see Section \ref{sbfs}). 
\end{theorem}
\begin{proof}
The  operator $X_{\mu}$  is build out of real-valued  polynomials of the multiplication and differentiation operator and the function valued coefficients of the operator, in particular the zero component, are continuously differentiable for all $\mathbf{p}\in\R^{3}$. Since,   the Schwartz space is analytic w.r.t. the representations of infinitesimal generators of the Heisenberg-Weyl group, it follows that the symmetric operator $X_{\mu}$ (see former Theorem) has a total set of analytic vectors from which essential self-adjointness follows, \cite[Theorem X.39]{RS2}. 
\end{proof}  
In the context of relativistic coordinate and conjugate operators, \textbf{the}  coordinate operator in relativistic QFT, is the Newton-Wigner-Pryce operator see \cite{NW49}, \cite{PR48}, \cite{SS} and \cite{Muc3}   and references therein.  For the one-particle case  (see  \cite{SS} and \cite[Lemma 3.3]{Muc3}) the operator is given by 
	\begin{equation}\label{NWP}
	(X_{j} \varphi)(\mathbf{p})=-i \left( \frac{p_j}{2\omega_{\mathbf{p}}^2}
	+   \frac{\partial}{\partial p^j } 
	\right)\varphi(\mathbf{p}),
	\end{equation}
By taking a closer look at the expression for the spatial part of $X$ we recognize  that the NWP-operator  is an essential part of the spatial Coleman-Mandula operator, i.e. 

	\begin{align*}(X_j\varphi)(\mathbf{p})&=	\frac{i}{m^2} 
	\bigg(
	\frac{3 \,p_{j}}{2 }  +  {p_jp^{k}} 	  \frac{\partial}{\partial p^k }- m^2\frac{\partial}{\partial p^j} \bigg)\varphi(\mathbf{p})\\&=
 	 \frac{i\,p_{j}}{2m^2} 
	\bigg(3+ 	\frac{ m^2\,}{  	\omega_{\textbf{p}}^2 }+  {2 p^{k}} 	  \frac{\partial}{\partial p^k }  \bigg)\varphi(\mathbf{p})+
(
	X_{j}^{NWP}\varphi)(\mathbf{p}).
	\end{align*}  
	This is an additional justification for the Coleman-Mandula operator to be called a generalized coordinate operator. In particular, it is the   additional terms to the   Newton-Wigner operator that regulate the breaking of covariance. \newline\newline
 Next, we calculate the adjoint action  w.r.t. the Poincar\'e group  on the Coleman-Mandula operator. 
\begin{theorem}\label{trafo}
	The adjoint action  w.r.t. proper orthochronous  Poincar\'e group $\mathscr{P}^{\uparrow}_{+}=\mathscr{L}^{\uparrow}_{+}\ltimes\mathbb{R}^4$  on the Coleman-Mandula operator is given by   
	 \begin{align}\label{et1}
 U(a, \Lambda)X_{\mu} U(a, \Lambda)^{-1}= (\Lambda^{T})_{\mu}^{\,\, \,\nu} \left( X_{\nu}+a_{\nu} -\frac{a^{\rho}}{M^2}P_{\rho}P_{\nu}\right).
	 \end{align} 
\end{theorem}
\begin{proof}
	First we investigate the action of the Lorentz transformations on the Coleman-Mandula operator. From the algebra it is clear that the conjugate operator transforms as, 
\begin{align*}
U(0,\Lambda)\,X_{\mu}U(0,\Lambda)^{-1}=
( \Lambda^{T} X)_{\mu} 
  ,
\end{align*}
 where we used the algebra between the momentum and the Lorentz generators (see \cite[Chapter 7, Equation 68]{Sch}) and replaced the symbol $P$ by $X$. The next adjoint action is a bit more involved, 
 \begin{align*}
 U(a,\bone)\,X_{\mu}U(a,\bone)^{-1}&= X_{\mu}+ia^{\rho} [P_{\rho},X_{\mu}]+\frac{i^2}{2!}a^{\rho}a^{\sigma} \underbrace{[P_{\rho},[P_{\sigma},X_{\mu}]]+\cdots}_{=0}\\&= 
 X_{\mu}+a_{\mu}  -\frac{1}{M^2}a^{\rho} P_{\rho}
 P_{\mu}
 ,
 \end{align*}
 where in the last lines the Baker-Campbell Hausdorff formula and the explicit algebra was used. Since a general Poincar\'e transformation can be written as the product, 
 \begin{align*}
 U(a,\bone) U(0,\Lambda)=U(a, \Lambda),
 \end{align*}
the proof is concluded.  
\end{proof}
 Hence, for particles where the mass is larger then the spatial velocity, the last term in the transformation vanishes and we obtain the well-known transformational behavior of the coordinate operator in quantum mechanics.   \\\\
 In this letter we obtained a generalized coordinate operator that generates a Snyder-like spacetime (see Equation (\ref{ss})), however the advantages of this, as we called it, physically realistic operator is that the commutation relations do not break the translation invariance as is usually done by the commutation relations of the Snyder spacetime. This is  seen by translating the coordinate operators (and assuming they translate covariant) on the left hand side of the equation,  
 \begin{align*}U(a,\bone)\,[X_{\mu},X_{\nu}]U(a,\bone)^{-1} =[X_{\mu}+a_{\mu},X_{\nu}+a_{\nu}]
 =[X_{\mu},X_{\nu}]
 \end{align*}
 where one uses the fact that the constant vectors commute. However, since the right hand side, i.e. the side depending on the Lorentz generator are not translation invariant an inconsistency appears. By using the new commutation relations of the Coleman-Mandula operator with the momentum we prove next that this problem is absent in our case.
 \begin{theorem}
 	The Snyder spacetime generated by the Coleman-Mandula operator (see Equation \ref{ss}) transforms covariant under the actions of the proper orthochronous  Poincar\'e group $\mathscr{P}^{\uparrow}_{+}$. 
 \end{theorem}
 \begin{proof}
 	We begin by using the former theorem to transform the left hand side of the commutator relations,  with $(a,\Lambda)\in\mathscr{P}^{\uparrow}_{+}$
 	\begin{align*}&U(a,\Lambda)\,[X_{\mu},X_{\nu}]U(a, \Lambda)^{-1} =  
 (\Lambda^{T})_{\mu}^{\,\,\,\,\sigma} (\Lambda^{T})_{\nu}^{\,\,\,\,\lambda}[
 	X_{\sigma}  -a^{\rho}\frac{1}{M^2}P_{\rho}P_{\sigma},
 	X_{\lambda}  -a^{\beta}\frac{1}{M^2}P_{\beta}P_{\lambda}] 	\\&= 	\frac{1}{M^2}
 		(\Lambda^{T})_{\mu}^{\,\,\,\,\sigma} (\Lambda^{T})_{\nu}^{\,\,\,\,\lambda}
 	\left( 
 		iJ_{\sigma\lambda}-a^{\beta}	[	X_{\sigma},P_{\beta}P_{\lambda}
 		]+\sigma\leftrightarrow\lambda \right)
 		\\&= 	\frac{i}{M^2}
 	(\Lambda^{T})_{\mu}^{\,\,\,\,\sigma} (\Lambda^{T})_{\nu}^{\,\,\,\,\lambda}
 		\left( 
 		J_{\sigma\lambda}+a_{\lambda}P_{\sigma}-a_{\sigma}P_{\lambda}  \right),
 	\end{align*}
 	where in addition to the former theorem we used the commutation relations of the momentum with the Coleman-Mandula operator. Now the right hand side of the commutator equation (\ref{ss}) transforms exactly in this fashion since it is merely the generator of Lorentz boosts and rotations. 
 \end{proof}
 The Snyder space that is generated by our operators is therefore not only Lorentz-covariant, as is the standard Snyder spacetime, but in addition it is covariant under translations. This is due to the extra term that appears   in the extended Heisenberg commutator relations  (\ref{cr1}). Hence, our special case of a Snyder algebra transforms covariant under Poincar\'e transformations, while the usual Snyder algebra given in  \cite{SN}  is not covariant under the translations as was   pointed out in \cite{Yang47}. Therefore, our algebra is a Poincar\'e invariant alternative to the translational non-covariant algebra of Snyder. 
  
\subsection{The soccer-ball problem solved}
Theories that deal with deformation of special relativity or deformations of  relativistic momentum space, induce modifications to the energy-momentum relations, see for example \cite[Formula 1]{AC}. Investigations   show  that these modifications lead to serious problems in the treatment of the multi-particle regime and they are even more serious in the macroscopic regime. This problem is known as the \textit{soccer-ball problem},  \cite{AC2}.  However, since our spacetime is generated such that it leaves the energy-momentum relations untouched, we do not encounter the soccer-ball problem that usually haunts such theories.  \newpage
\section{Extending Snyder-Spacetime by Deformation}
Recent progress in deformation theory, provided us with a mathematical well-defined and physically well motivated way of obtaining a quantum spacetime, i.e. a noncommutative spacetime, from a given set of coordinate operators. In \cite{AA} and \cite{Muc1} the noncommutative Moyal-Weyl space was obtained by deformation quantization of the  ordinary quantum mechanical coordinate operator, by using as generators of the deformation the translation operators, i.e. the momentum operators. The result was a  noncommutative space that is connected to  the well-known example of the Landau-quantization. The effect was reformulated in the language of deformation quantization and in addition applied to the gravito-magnetic field resulting in the so called gravito-magnetic Landau-quantization and the gravito-magnetic Zeemann-effect.\\\\
  Following the quantum mechanical example of the Landau quantization in the context of deformation quantization we deform in this section   the Coleman-Mandula operator in order to see additional corrections coming from a non-vanishing constant Moyal-Weyl. In particular the question that is answered in this context is the following: Can we obtain a noncommutative  spacetime in a addition to the Snyder type one, by a strict deformation quantization? And if so, will there be terms resembling additional relativistic corrections like in example for \cite{Muc3}? To answer these  questions, we need to evaluate the   commutator of  the deformed Coleman-Mandula operators. The generators of deformation are in this context (as in  \cite{AA}, \cite{Muc1} and \cite{Muc3}) the momentum operators.   Moreover, in order to preform the deformation quantization we need to prove that the deformation of these unbounded operators is well-defined. The more physical interested reader can skip the following subsection.
\subsection{Mathematical Properties of the Deformation}
We start by giving the definition  of  the deformation known by warped convolutions, (see \cite{RI} and \cite{BLS}), 
 \begin{definition}\label{defwca}
 	Let $\Theta$ be a skew-symmetric matrix on $\mathbb{R}^{d}$ and   $\mathcal{D}$ be the dense and  stable subspace of vectors in $\mathcal{H}$, which transform smoothly under the unitary operator $U$. Finally, let   $E$ be the spectral resolution of the   operator $U$. Then, the warped convolutions  of an operator $A$, defined on  $\mathcal{D}\subset\mathcal{H}$ and   denoted by $A^{\Theta}$, is  defined according to
 	\begin{equation}\label{WC1}
 	A^{\Theta} :=\int dE(x)\,\alpha_{\Theta x}(A) = (2\pi)^{-d}
 	\lim_{\epsilon\rightarrow 0}
 	\iint dx\, dy \,\chi(\epsilon x,\epsilon y )\,e^{-ixy}\, \alpha_{\Theta x}(A)\, U(y) ,
 	\end{equation}
 	where $\chi \in\mathscr{S}(\mathbb{R}^d\times\mathbb{R}^d)$ with $\chi(0,0)=1$ and $\alpha$ denotes the $\mathbb{R}^d$-action, i.e.
 	$$
 	\alpha_p(A):=U(p)\, A\,U(p)^{-1}\qquad  \forall p\in\mathbb{R}^{d}.
 	$$
 \end{definition} 
The former definition   is  rigorously defined  for a certain   bounded operators, see 
\cite{BLS}. However, since we are dealing with unbounded operators we   use the results of \cite{LW} in order to prove that the deformation of the Coleman-Mandula operator is well-defined. Hence, let us
consider the deformed operator $A^{{\Theta}}$ as follows
\begin{align*}
\langle \Psi,A^{{\Theta }}\Phi\rangle&=
(2\pi)^{-d}
\lim_{\epsilon\rightarrow 0}
\iint  \, dx \,  dy \, e^{-ixy}  \, \chi(\epsilon x,\epsilon
y) {\langle \Psi, 
	U(y)\alpha_{\Theta x}(A)\Phi\rangle} \nonumber \\&
=:
(2\pi)^{-d}
\lim_{\epsilon\rightarrow 0}
\iint  \, dx\,  dy \, e^{-ixy}  \, \chi(\epsilon x,\epsilon
y) \, b(x,y)
\end{align*}
for $\Psi, \Phi \in \mathcal{D}^{\infty}(A):=\{
\varphi \in\mathcal{D}(A)| U(x)\varphi \in \mathcal{D}(A)$  is smooth in $\|\cdot\|_{\mathcal{H}}\}$ and 	   $\chi \in\mathscr{S}(\mathbb{R}^d\times\mathbb{R}^d)$ with $\chi(0,0)=1$. Note that   the scalar product is   w.r.t. $\mathcal{H}$.
The expression is well-defined if $b(x,y)$ is a symbol which is given in the following definition (\cite[Section 7.8, Definition 7.8.1]{H}).

 	\begin{definition}
 		Let $r$, $\rho$, $\delta$, be real numbers with $0<\rho\leq1$ and $0\leq \delta <1$. Then we
 		denote by $S^{r}_{\rho,\delta}(X\times \mathbb{R}^d)$, the set of all $b\in C^{\infty}(X\times
 		\mathbb{R}^d)$ such that for every compact set $K\subset X$ and all  $\gamma, \kappa$  the estimate 
 		\begin{equation*}
 		|\partial_{x}^{\gamma}\partial_{y}^{\kappa}b(x,y)|\leq
 		C_{\gamma,\kappa,K}(1+|x|)^{r-\rho|\gamma|+\delta|\kappa|},\qquad x\in K, \,\, y\in
 		\mathbb{R}^d,  
 		\end{equation*}
 		is valid for some constant $C_{\gamma,\kappa,K}$. The elements  $S^{r}_{\rho,\delta}$ are called
 		symbols of order $r$ and type $\rho,\delta$. Note that  $\gamma, \kappa$ are multi-indices and  $|\gamma|$, $|\kappa|$ are the corresponding sums of the index-components. 
 	\end{definition}

 	\begin{lemma}\label{lfhs}
 		Assume that the derivatives of the adjoint action of $A$ w.r.t. the unitary operator $U$ are polynomially bounded on vectors in $\mathcal{D}^{\infty}(A)$, i.e.
 		\begin{equation}\label{pb}
 		\|\partial_{x}^{\gamma}\alpha_{\Theta x}(A)\Phi\| \leq  C_{\gamma}(1+|x|)^{r-\rho|\gamma|},\qquad \forall \Phi \in
 		\mathcal{D}^{\infty}(A).
 		\end{equation}
 		Then, $b(x,y)$ belongs to the symbol class $S^{r}_{\rho,0}$ for $\Psi, \Phi \in
 		\mathcal{D}^{\infty}(A)$ and therefore the deformation  of the unbounded operator $A$ is given as a well-defined
 		oscillatory integral.
 	\end{lemma} 
 	\begin{proof} For proof see \cite{Muc5} or    \cite{LW} and \cite[Theorem 1]{AA}. 
 	\end{proof}
Basically, what the former Lemma   implies is the following. If the deformed operator is polynomially bounded w.r.t. the unitary action the deformation, given by oscillatory integrals, is well-defined. For the following deformation we work in four dimensions which restricts the operators to be defined on vectors in  $\mathscr{S}(\mathbb{R}^3)$.

\begin{lemma}\label{lfhs}
Derivatives of the adjoint action of $X_{\mu}$, for all $\mu$, w.r.t. the unitary operator $U$, generated by the infinitesimal generators of the translations,  are polynomially bounded on vectors in $\mathscr{S}(\mathbb{R}^3)$, i.e.
	\begin{equation}\label{pb}
	\|\partial_{x}^{\gamma}\alpha_{\Theta x}(X_{\mu})\Phi\| \leq  C_{\gamma}(1+|x|)^{1-|\gamma|}, \qquad \forall \Phi \in \mathscr{S}(\mathbb{R}^3).
	\end{equation}
	Therefore the  warped convolution, i.e. the deformation, of this operator is well-defined. 
\end{lemma} 
\begin{proof}In the following proof we use the explicit adjoint action of the Coleman-Mandula operator w.r.t. the translations see  
	\begin{align*} 
	\|\partial_{x}^{\gamma}\alpha_{\Theta x}(X_{\mu})\Phi\|&
	=	\|\partial_{x}^{\gamma}	( X_{\mu}+( \Theta x)_{\mu}  -\frac{( \Theta x)^{\rho}}{m^2} 
	P_{\mu}P_{\rho})\Phi\|\\&
	=	\|\partial_{x}^{\gamma}	( X_{\mu}+x_{\lambda}  P^{\lambda }_{\mu}) \Phi\| ,
	\end{align*}
	where the matrix $\Theta$ is a skew-symmetric matrix on $\R^4$ and we defined the matrix valued-operator $P^{\lambda }_{\mu}:=\Theta_{\mu}^{\lambda}-\Theta^{\rho\lambda}P_{\mu}P_{\rho}$. 
	For $\gamma=0$ we have,
	\begin{align*} 
		\| 	( X_{\mu}+x_{\lambda}  P^{\lambda }_{\mu}) \Phi\|&
 \leq \underbrace{	\| 	  X_{\mu}  \Phi\|}_{=:F_1}+\|  x_{\lambda}  P^{\lambda }_{\mu}  \Phi\|\\&
 \leq 	F_1+ |  x   |\underbrace{\| P^{\lambda }_{\mu}  \Phi\|}_{=:F_2} \leq
 C_0(1+ |  x   |)
	\end{align*}
where the constants $F_1,F_2$ introduced are finite due to the choice of the domain. From the  finiteness of the constants $F_1,F_2$ it follows that a finite constant $C_0$ can be found such that the inequality is fulfilled. For $\gamma=1$ the $x$-dependence vanishes, i.e. we have
	\begin{align*} 
	\| 	    P^{\lambda \gamma_1 }   \Phi\|&
	\leq C_1,
	\end{align*}
 and for higher $\gamma$'s the whole expression is zero. 
\end{proof}
\begin{lemma}
From the former Lemma it follows that the deformed Coleman-Mandula operator $X^{\Theta}_{\mu}$ is essential self-adjoint on the dense and stable domain $\mathscr{S}(\mathbb{R}^3)$. 
\end{lemma}
\begin{proof}
	The proof for general symbols can be found in \cite{Muc5}.
\end{proof}
\subsection{Deformed Snyder Space}
   In this section we first give the explicit result of the deformation on the Coleman-Mandula operator and afterwards we take the commutator for different components in order to obtain a new noncommutative spacetime. Moreover, we examine the transformation property of the operator under Poincar\'e -transformations
   \begin{theorem}
   	The explicit result of the deformation of the Coleman-Mandula operator, i.e.  $X_{\mu}^{\Theta}$, is given on vectors of  the  domain $\mathscr{S}(\mathbb{R}^3)$ as
   	\begin{align}
X_{\mu}^{\Theta}=X_{\mu}+( \Theta P)_{\mu} . 
   	\end{align}
   	The \textbf{noncommutative spacetime} that the deformed  Coleman-Mandula operator spans is given by,
   	   	\begin{align}
   	[X_{\mu}^{\Theta},X_{\nu}^{\Theta}]&=	i\frac{J_{\mu\nu}}{M^2} -2i   \Theta_{\mu\nu} - \frac{2i}{M^2} \left((\Theta P)_{\mu}P_{\nu}-(\Theta P)_{\nu}P_{\mu}\right). 
   	\end{align} In terms of the relativistic velocity operator $V_{\mu}=P_{\mu}P_{0}^{-1}$ the noncommutative spacetime reads,
   	 	\begin{align}
   	[X_{\mu}^{\Theta},X_{\nu}^{\Theta}]=
   	i\frac{J_{\mu\nu}}{M^2}-2i   \Theta_{\mu\nu} - 2i\frac{P_{0}^2}{M^2} \left((\Theta V)_{\mu}V_{\nu}-(\Theta V)_{\nu}V_{\mu}\right)
   	.
   	   	\end{align} 
   \end{theorem}
\begin{proof}
	The deformation is given by the well-defined spectral integral,
	   	\begin{align*}
	\int dE(x)\,\alpha_{\Theta x}(X_{\mu})&=
	\int dE(x)\, (	 X_{\mu}+( \Theta x)_{\mu}  -\frac{( \Theta x)^{\rho}}{M^2} 
	P_{\mu}P_{\rho}) \\&=
  X_{\mu}+( \Theta P)_{\mu}  -\frac{( \Theta P)^{\rho}}{M^2} 
	P_{\mu}P_{\rho} \\&= X_{\mu}+( \Theta P)_{\mu} ,
	   	\end{align*}
	where in the last lines we used the transformational behavior of the Coleman-Mandula operator under translations  (see Theorem \ref{trafo}) and the skew-symmetry of the deformation matrix $\Theta$. Next we calculate the commutator, 
	\begin{align*}
 [X_{\mu}^{\Theta},X_{\nu}^{\Theta}]&= [X_{\mu}+( \Theta P)_{\mu} , X_{\nu}+( \Theta P)_{\nu} ]\\&=
 [X_{\mu} , X_{\nu}  ]+ [X_{\mu} ,  ( \Theta P)_{\nu} ]+ [ ( \Theta P)_{\mu} , X_{\nu}  ]
 \\&=i\frac{J_{\mu\nu}}{M^2}+ \left(\Theta_{\nu}^{\,\,\lambda}[X_{\mu} , P_{\lambda} ] -\mu\leftrightarrow\nu\right)
  \\&= i\frac{J_{\mu\nu}}{M^2} +i \left(\left( \Theta_{\nu\mu} -\frac{(\Theta P)_{\nu}P_{\mu}}{M^2}\right)-\mu\leftrightarrow\nu\right).
	\end{align*}

\end{proof} 
The resulting spacetime displays interesting features. In addition to the Snyder-spacetime we obtained 
a Moyal-Weyl term, i.e. a constant non-commutativity and furthermore terms depending on the Moyal-Weyl Matrix $\Theta$ and quadratic terms in the momenta.\newline\newline Moreover, the spacetime that we obtained by deformation of a specific Snyder-operator can be understood as a covariant extension of the space-space non-commutativity obtained in \cite[Theorem 3.1]{Muc3}, i.e. 
	\begin{align*}
  [Q_{i}^{  \Theta},Q_{j}^{  \Theta}]=-2 i\left(
 \Theta_{0i} V_j /c-\Theta_{0j} V_i/c \right) 
 -2i\Theta_{ij} ,
	\end{align*}
	where $Q$ represents the Newton-Wigner and  $V$  is the spatial part of the relativistic velocity operator. Moreover, we called it a covariant extension of the former model since the objects we deformed are covariant. By comparing the expressions and assuming that $M$ is large compared to the spatial momentum the former spacetime leads to the latter.

\subsection{Covariance of the Deformed Snyder Space}
In this section we discuss the covariance properties of the deformed Snyder spacetime under transformations of the Poincar\'e-group. However, before, we give the Poincar\'e-transformed commutator relations we discuss the Moyal-Weyl spacetime and the apparent breaking of relativistic covariance. The Moyal-Weyl spacetime is a  noncommutative constant  spacetime and the algebra   is generated by self-adjoint operators (on some Hilbert space) $\hat{x}$ that satisfy
$$[\hat{x}_{\mu},\hat{x}_{\nu}]=i\Theta_{\mu\nu} .
$$
The implicit breaking of the covariance of the Moyal-Weyl spacetime can be proven as follows. Assume that the coordinate operators that generate the noncommutative spacetime transform covariantly, i.e. 
$\hat{x}_{\mu} \rightarrow (\Lambda^{T}\hat{x})_{\mu}$. By applying this transformation to both sides of the commutator relation of the Moyal-Weyl it follows that  for covariance to hold  we have $$\Theta=\Lambda^{T} \Theta \Lambda .$$ Since $\Theta$ is constant this can only hold if $\Theta=0$, or by giving this transformation an appropriate interpretation in the context of quantum field theory, \cite{GL1}, \cite{GL2} or quantum groups \cite{CH}. Hence, the apparent breaking is generated by the different transformation of the two-sides of the equation. In the present case   the equality respects the transformation since the deformed generalized coordinate operators were not merely imposed but constructed.  \\ \\ In order to   examine the transformation
properties of the deformed operators the following proposition is used, \cite{BLS}.  
\begin{proposition}\label{blsp1}
	Let $V$ be a unitary or anti-unitary operator on $\mathscr{H}$ such that
	$VU(x)V^{-1}=U(\Lambda x)$, $x\in
	\mathbb{R}^d$, for some invertible matrix $\Lambda$. Then,  for an operator $A$ we have
	\begin{equation}
	VA^{\Theta}V^{-1}=(VAV^{-1})^{\sigma \Lambda \Theta \Lambda^{T}},
	\end{equation}
	where $\Lambda^{T}$ is the transpose of $\Lambda$ w.r.t the chosen bilinear form, $\sigma=1$ if $V$ is
	unitary
	and $\sigma=-1$ if $V$ is anti-unitary.
\end{proposition}

 By using the former proposition we are able to prove the following result.
\begin{theorem}
	Since the Coleman-Mandula operator $X_{\mu}$ transforms covariantly under transformations of the  proper orthochronous  Lorentz group (see Theorem \ref{trafo}), the deformed version, i.e.  $X_{\mu}^{\Theta}$, transforms by using the former Proposition as follows,
		\begin{equation}
	 U(0,\Lambda)\,X_{\mu}^{\Theta}\, U(0,\Lambda)^{-1}=(\Lambda^{T} 
	 X)_{\mu}^{  \Lambda \Theta \Lambda^{T}}.
		\end{equation}
		Moreover, by adding the translations to Lorentz-transformation the operator transforms as, 
		\begin{align}\label{t67}
		U(a,\Lambda) X_{\mu}^{\Theta}	U(a,\Lambda)^{-1}&= (\Lambda^{T} 
		X)_{\mu}^{  \Lambda \Theta \Lambda^{T}}+(\Lambda^{T} a-\frac{1}{M^2}(a\cdot P)\Lambda^{T}  P )_{\mu}
		\end{align}
		This implies that  the commutator relations of the deformed Snyder spacetime transform  \textbf{ covariant  under Poincar\'e-transformations}.
  
\end{theorem}

\begin{proof}
	The covariant transformation behavior of the Coleman-Mandula operator under Lorentz-transformations  is known from Theorem \ref{trafo}. In order to prove the transformational manner of the deformed operator we use the former proposition and Theorem \ref{trafo}, i.e.
	\begin{equation*}
	U(0,\Lambda)\,X_{\mu}^{\Theta}\, U(0,\Lambda)^{-1}=
	( U_{\Lambda}\,X_{\mu} \, U_{\Lambda}^{-1})^{  \Lambda \Theta \Lambda^{T}}=(\Lambda^{T} 
	X)_{\mu}^{  \Lambda \Theta \Lambda^{T}}.
	\end{equation*}
	Therefore the left hand side of the commutator relations of the Snyder spacetime transform as,
		\begin{align*}
&	[ (\Lambda^{T} 
		X)_{\mu}^{  \Lambda \Theta \Lambda^{T}}, (\Lambda^{T} 
		X)_{\nu}^{  \Lambda \Theta \Lambda^{T}}] =(\Lambda^{T})^{\,\,\,\, \kappa}_{\mu} (\Lambda^{T})^{\,\,\,\,\lambda}_{\nu} 
			[  			X_{\kappa}^{  \Lambda \Theta \Lambda^{T}}, 
			X_{\lambda}^{  \Lambda \Theta \Lambda^{T}}]
\\		& 	=(\Lambda^{T})^{\,\,\,\, \kappa}_{\mu} (\Lambda^{T})^{\,\,\,\,\lambda}_{\nu} \left( i\frac{J_{\kappa\lambda}}{M^2} -2i   ( \Lambda \Theta \Lambda^{T})_{\kappa\lambda} - \frac{2i}{M^2} \left((( \Lambda \Theta \Lambda^{T})P)_{\kappa}P_{\lambda}-(( \Lambda \Theta \Lambda^{T}) P)_{\lambda}P_{\kappa}\right) \right)\\&
=
  \left( i\frac{(\Lambda^T J \Lambda)_{\mu\nu}}{M^2} -2i   \Theta_{\mu\nu} - \frac{2i}{M^2} \left((   \Theta  (\Lambda^T P))_{\mu}(\Lambda^T P)_{\nu}-( \Theta   (\Lambda^T P))_{\nu}(\Lambda^T P)_{\mu}\right) \right),
		\end{align*}
	which is equivalent to the transformation  of the right-hand side.  Next, we apply the translations to the left-hand side  of the Snyder spacetime that is generated by the deformed Coleman-Mandula operator. Before, doing so let us study the action of the translation on the deformed Coleman-Mandula operator, 
	
	\begin{align*}
	U(a,\bone)X_{\mu}^{\Theta}	U(a,\bone)^{-1}&=
		U(a,\bone)X_{\mu}U(a,\bone)^{-1}+(\Theta P)_{\mu}\\&=
		X_{\mu}+a_{\mu} -\frac{1}{M^2}a^{\rho}P_{\rho}P_{\mu}+(\Theta P)_{\mu}
		\\&=X_{\mu}^{\Theta}+a_{\mu} -\frac{1}{M^2}a^{\rho}P_{\rho}P_{\mu},
	\end{align*}
	where in the last line we used the commutativity of the momentum operator among its components and the transformation behavior  (see Equation (\ref{et1})) of the Coleman-Mandula operator under translations.  
	Therefore the left hand side of our physically realistic Snyder spacetime transforms as  
		\begin{align*}
	 	U(a,\bone)[X_{\mu}^{\Theta},X_{\nu}^{\Theta}]	U(a,\bone)^{-1}&=
	 [	X_{\mu}^{\Theta}+a_{\mu} -\frac{1}{M^2}a^{\rho}P_{\rho}P_{\mu},X_{\nu}^{\Theta}+a_{\nu} -\frac{1}{M^2}a^{\sigma}P_{\sigma}P_{\nu}]\\&
	= [	X_{\mu}^{\Theta},X_{\nu}^{\Theta}]-\frac{1}{M^2}a^{\rho}[P_{\rho}P_{\mu},X_{\nu} ]+\mu\leftrightarrow\nu\\&
	= [	X_{\mu}^{\Theta},X_{\nu}^{\Theta}]+\frac{i}{M^2}\left(  a_{\nu}P_{\mu}-a_{\mu}P_{\nu}\right),
		\end{align*}
This is the exact form in which the right hand side of the commutator relation transforms under translations.  Since a general Poincar\'e transformation can be written as the product, 
 \begin{align*}
U(a,\bone) U(0,\Lambda)=U(a, \Lambda),
\end{align*}
	we only need to use the former results in order to prove that the Snyder spacetime, generated by the Coleman-Mandula operators, is Poincar\'e invariant. 
	 
\end{proof}  
The former proposition proves that there is no inconsistency in assuming (and in fact proving) the covariance of the coordinate operators and the constant $\Theta$ as in the case of the Moyal-Weyl spacetime. Hence, the constant non-commutativity, given by $\Theta$, looks the same in all Lorentz-frames. Let us briefly discuss and compare the former result with results obtained by using a twisted  Poincar\'e symmetry as in \cite{CH}. 
\\\\The authors in the formentioned  paper argue that while the algebraic point of view (i.e. by generating the Moyal-Weyl spacetime using operators)   highlights the violation of the Lorentz group, by using functions and the  twisted co-product, which is a deformed co-product that satisfies the twist equation (see \cite[Equation 2.4]{CH})  the noncommutative spacetimes becomes invariant  under transformations of the twist-deformed Poincar\'e group. Hence, instead of using the undeformed Poincar\'e group one takes the twist as well in the group into account in order for the functions (in this case the coordinates) to transform invariant. However, our approach uses the undeformed  Poincar\'e-group, while it changes the commutator relations between the coordinate operators.   
\section{Discussion}
In this work, we motivated from different perspectives an operator which,   takes in a relativistic context the rightful place of the coordinate operator. In particular, this self-adjoint operator is noncommutative and generates a Snyder-type spacetime that is not plagued from the problems, as the breaking of translation covariance or the soccer-ball problem,  that the Snyder spacetime usually suffers from. 
\newline\newline
In addition to representing this operator on the Fock-space by using the massive scalar field, we applied Rieffel deformations to it, in order to generate a more complex noncommutative spacetime. The deformation of the Coleman-Mandula operator is proven to exist and the result is as well a self-adjoint object.
The result of the deformation was, in addition to the Snyder-type spacetime, a  Moyal-Weyl type of noncommutative spacetime with velocity dependent terms. In the low energy limit it corresponds to the spacetime found in \cite{Muc3}.
\newline\newline
The extraordinary property of this spacetime is the covariance under Poincar\'e transformations. In the well-known noncommutative models  that are applied to physics, one usually encounters either a problem with the Lorentz-covariance or the translational covariance. For example, the Moyal-Weyl spacetime, i.e. 
$$[X_{\mu},X_{\nu}]=i\Theta_{\mu\nu}\bone$$
where $\Theta$ is a constant skew-symmetric matrix, is covariant (more specifically invariant) under global translations, i.e. $X_{\mu}\rightarrow X_{\mu}+a_{\mu}\bone$, but not covariant under Lorentz transformations. The Snyder spacetime on the other hand,  
$$[X_{\mu},X_{\nu}]=i \kappa^2 J_{\mu\nu}$$
where $\kappa$ being a physical dimension-full constant, is covariant w.r.t. to Lorentz transformations, i.e. $X_{\mu}\rightarrow(\Lambda^{T}  X)_{\mu}$, but not covariant under translations since the right-hand side transforms but the left-hand side does not. However, by choosing the physically motivated operator properly and
deforming it we united both interesting spacetimes and achieved Poincar\'e
covariance, using the transformation given in Equation (\ref{t67}).

\bibliographystyle{alpha}
\bibliography{allliterature1}

\end{document}